\RequirePackage[l2tabu, orthodox]{nag} 

\documentclass{amsart}

\pdfoutput=1

\usepackage[utf8]{inputenc} 
\usepackage[colorlinks, citecolor=red, breaklinks=true]{hyperref} 
\usepackage[stretch=10, shrink=10]{microtype} 
\usepackage{todonotes} 
\usepackage{cases}
\usepackage{amscd, amssymb}  
\usepackage{commath} 
\usepackage{mathtools} 
\usepackage{doi}
\usepackage{xcolor}
\usepackage[initials,msc-links,non-compressed-cites,nobysame]{amsrefs}[2007/10/22]
\usepackage{subcaption}
\usepackage{paralist,enumitem}
\setlist{listparindent=0pt,parsep=3pt}
\setdefaultenum{1.}{}{}{}

\newcommand{\TitleWithUrl}[1]{\IfEmptyBibField{doi}%
  {\IfEmptyBibField{url}{\textit{#1}}%
    {\IfEmptyBibField{eprint}{\href {\BibField{url}}{\textit{#1}}}{\textit{#1}}}%
    }%
  {\href {https://doi.org/\BibField{doi}}{\textit{#1}}}}
\renewcommand{\eprint}[1]{\IfEmptyBibField{url}{\url{#1}}%
  {\href {\BibField{url}}{#1}}}

\BibSpec{article}{%
    +{}  {\PrintAuthors}                {author}
    +{,} { \TitleWithUrl}               {title}
    +{.} { }                            {part}
    +{:} { \textit}                     {subtitle}
    +{,} { \PrintContributions}         {contribution}
    +{.} { \PrintPartials}              {partial}
    +{,} { }                            {journal}
    +{}  { \textbf}                     {volume}
    +{}  { \PrintDatePV}                {date}
    +{,} { \issuetext}                  {number}
    +{,} { \eprintpages}                {pages}
    +{,} { }                            {status}
    +{,} { available at arXiv:\eprint}        {eprint}
    +{}  { \parenthesize}               {language}
    +{}  { \PrintTranslation}           {translation}
    +{;} { \PrintReprint}               {reprint}
    +{.} { }                            {note}
    +{.} {}                             {transition}
    +{}  {\SentenceSpace \PrintReviews} {review}
}

\BibSpec{collection.article}{%
    +{}  {\PrintAuthors}                {author}
    +{,} { \TitleWithUrl}                     {title}
    +{.} { }                            {part}
    +{:} { \textit}                     {subtitle}
    +{,} { \PrintContributions}         {contribution}
    +{,} { \PrintConference}            {conference}
    +{}  {\PrintBook}                   {book}
    +{,} { }                            {booktitle}
    +{,} { \PrintDateB}                 {date}
    +{,} { pp.~}                        {pages}
    +{,} { }                            {status}
    +{,} { available at \eprint}        {eprint}
    +{}  { \parenthesize}               {language}
    +{}  { \PrintTranslation}           {translation}
    +{;} { \PrintReprint}               {reprint}
    +{.} { }                            {note}
    +{.} {}                             {transition}
    +{}  {\SentenceSpace \PrintReviews} {review}
}
\BibSpec{book}{%
    +{}  {\PrintPrimary}                {transition}
    +{,} { \TitleWithUrl}               {title}
    +{.} { }                            {part}
    +{:} { \textit}                     {subtitle}
    +{,} { \PrintEdition}               {edition}
    +{}  { \PrintEditorsB}              {editor}
    +{,} { \PrintTranslatorsC}          {translator}
    +{,} { \PrintContributions}         {contribution}
    +{,} { }                            {series}
    +{,} { \voltext}                    {volume}
    +{,} { }                            {publisher}
    +{,} { }                            {organization}
    +{,} { }                            {address}
    +{,} { \PrintDateB}                 {date}
    +{,} { }                            {status}
    +{}  { \parenthesize}               {language}
    +{}  { \PrintTranslation}           {translation}
    +{;} { \PrintReprint}               {reprint}
    +{.} { }                            {note}
    +{.} {}                             {transition}
    +{}  {\SentenceSpace \PrintReviews} {review}
}

\BibSpec{thesis}{%
    +{}  {\PrintAuthors}                {author}
    +{.} { \PrintDate}                  {date}
    +{.} { \TitleWithUrl}               {title}
    +{:} { \textit}                     {subtitle}
    +{,} { \PrintThesisType}            {type}
    +{,} { }                            {organization}
    +{,} { }                            {address}
    +{,} { \eprint}                     {eprint}
    +{,} { }                            {status}
    +{}  { \parenthesize}               {language}
    +{}  { \PrintTranslation}           {translation}
    +{;} { \PrintReprint}               {reprint}
    +{.} { }                            {note}
    +{.} {}                             {transition}
    +{}  {\SentenceSpace \PrintReviews} {review}
}

\newtheorem{theorem}{Theorem}[section]

\newtheorem{proposition}[theorem]{Proposition}
\newtheorem{fact}[theorem]{Fact}

\theoremstyle{definition}
\newtheorem{definition}[theorem]{Definition}

\theoremstyle{remark}
\newtheorem{remark}[theorem]{Remark}

\numberwithin{equation}{section}

\newcommand{\cratio}{\mathrm{cr}}
\definecolor{darkblue}{rgb}{0,0,0.7}

\newsavebox\mybox

\title{Discrete mKdV equation via Darboux transformation}

\author{Joseph Cho}
\address[Joseph Cho]{Institute of Discrete Mathematics and Geometry, TU Wien, Wiedner Hauptstrasse 8-10/104, 1040 Wien, Austria}
\email{jcho@geometrie.tuwien.ac.at}
 
\author{Wayne Rossman}
\address[Wayne Rossman]{Department of Mathematics, Graduate School of Science, Kobe University, 1-1 Rokkodai-cho, Nada-ku, Kobe 657-8501, Japan}
\email{wayne@math.kobe-u.ac.jp}

\author{Tomoya Seno}
\address[Tomoya Seno]{Department of Mathematics, Graduate School of Science, Kobe University, 1-1 Rokkodai-cho, Nada-ku, Kobe 657-8501, Japan}
\email{tseno@math.kobe-u.ac.jp}

\date{}
\makeatletter
\@namedef{subjclassname@2020}{\textup{2020} Mathematics
  Subject Classification}
\makeatother

\begin{document}

\begin{abstract}
	We introduce an efficient route to obtaining the discrete potential mKdV equation emerging from a particular discrete motion of discrete planar curves.
\end{abstract}

\maketitle

\section{Introduction}

In this work, we illuminate the relationship between the discrete and semi-discrete potential modified Korteweg-de Vries (mKdV) equations on the one hand, and on the other, the transformation theory of the smooth potential mKdV equation via the use of infinitesimal Bianchi cubes of Darboux transformations keeping arc-length polarization.

The works \cites{matsuura_discrete_2012-1, inoguchi_motion_2012} used the discrete frames of discrete motions of a discrete curve, to produce the discrete potential mKdV equation as the compatibility condition; rather, \cites{kaji_linkage_2019, inoguchi_explicit_2012} considered the analogous result for continuous motions to obtain the semi-discrete potential mKdV equation.
However, the discrete and semi-discrete potential mKdV equations appear together in the context of Bäcklund transformations and permutability of the smooth potential mKdV equation \cite{wadati_backlund_1974}.
This suggests that a suitable definition of a 3-dimensional integrable lattice with two discrete parameters and one smooth parameter would allow one to combine both the discrete and continuous motions of a discrete curve, resulting in a unified approach to obtaining the discrete and semi-discrete potential mKdV equations.

An integrable lattice involving one smooth and two discrete parameters appeared in \cite{burstall_semi-discrete_2016}, which investigated Darboux transformations and permutability of smooth polarized curves, leading to the \emph{semi-discrete isothermic surfaces} of \cite{muller_semi-discrete_2013} and their Darboux transformations.
Darboux transformations of smooth polarized curves have been given a new interpretation in \cite{cho_infinitesimal_2020} as \emph{Darboux deformations} of discrete polarized curves, describing a motion of discrete curves that results in the semi-discrete isothermic surfaces. Characterizing the continuous motion of discrete curves in \cite{kaji_linkage_2019, inoguchi_explicit_2012} as Darboux deformations keeping discrete arc-length polarization, this work \cite{cho_infinitesimal_2020} simplified the process of obtaining the semi-discrete mKdV equation by interpreting the potential function as the tangential angle -- the turning angle of the tangent vector measured from a fixed axis --  of the smooth curves in the system.
Thus, the 3-dimensional integrable lattice of Darboux transformations and permutability of smooth polarized curves in \cite{burstall_semi-discrete_2016} presents itself as a suitable choice for unifying the discrete and semi-discrete systems yielding the respective potential mKdV equations.

In this paper, we reinterpret such a 3-dimensional system as the permutability between Darboux transformations and Darboux deformations of a discrete polarized curve, a notion we refer to as \emph{infinitesimal Bianchi cube}.
To do this, we explicitly state the expected definition of Darboux transformations of a discrete polarized curve in Definition~\ref{def:dFlow}, expected since the image of successive Darboux transformations should yield discrete isothermic surfaces as defined in \cite{bobenko_discrete_1996-1}.
Then, we further look at the integrable reduction of this system, by imposing the arc-length polarization on the discrete curves, and show that the Darboux deformations and transformations keeping the arc-length polarizations permute as well (see Proposition \ref{PermOfDarb}).

We then show that such integrable reduction yields a 1-parameter family of solutions to the discrete potential mKdV equation (see Theorem~\ref{thm1}) using the well-known calculations dating back to Bianchi \cite{bianchi_sulla_1892} (see also \cites{wadati_backlund_1974, rogers_backlund_2002}), and make connection to the discrete motion of discrete curves introduced in \cites{matsuura_discrete_2012-1, inoguchi_motion_2012} that also results in the discrete potential mKdV equation.

\section{Preliminaries}
In this section, we recall the definitions and results from \cite{burstall_semi-discrete_2016, cho_infinitesimal_2020} directly related to this paper.
For this, let $s \in I$ where $I \subset \mathbb{R}$ is a smooth interval in the reals, polarized by a non-vanishing quadratic differential $\frac{\dif s^2}{m}$.

\begin{definition}[{\cite[Definition 2.4]{burstall_semi-discrete_2016}}]
	A pair of smooth polarized curves $f_0, f_1 : (I, \frac{\dif s^2}{m}) \to \mathbb{C}$ is called a \emph{Darboux pair} with parameter $\mu$ if
	\[
		\frac{f_0' f_1'}{(f_0 - f_1)^2} = \frac{\mu}{m}.
	\]
	If two curves are a Darboux pair, then we call one curve a \emph{Darboux transformation} of the other.
\end{definition}
Darboux transformations of smooth polarized curves obey Bianchi permutability:
\begin{fact}[{\cite[Theorem 2.8]{burstall_semi-discrete_2016}}]\label{fact:perm}
	Starting with a polarized curve $f_0$, and its Darboux transforms $f_1$ and $f_2$ with parameters $\mu_1$ and $\mu_2$, respectively, one has a fourth curve $f_{12}$ that is simultaneously a Darboux transform of $f_1$ and $f_2$ with parameters $\mu_2$ and $\mu_1$, respectively, where the fourth curve is determined algebraically via the cross-ratio condition
	\begin{equation}\label{eqn:crsmooth}
		\cratio (f_0, f_1, f_{12}, f_2) = \frac{\mu_2}{\mu_1}.
	\end{equation}
\end{fact}
Therefore, successive Darboux transformations of smooth polarized curves result in the semi-discrete isothermic surfaces; the Bianchi permutability of smooth polarized curves then gives the Darboux transformation of semi-discrete isothermic surfaces.
We also recall:
\begin{definition}[{\cite[p.\ 45]{burstall_semi-discrete_2016}}]
	A smooth polarized curve $f_0: (I, \frac{\dif s^2}{m}) \to \mathbb{C}$ is \emph{arc-length polarized} if
	\[
		\frac{\dif s^2}{m} = |{\dif f_0}|^2.
	\]
\end{definition}
It is then readily seen that if an arc-length polarized curve $f_0$ is also arc-length parametrized, then we have $m \equiv 1$.
With the arc-length polarizations available, Darboux transformations keeping the arc-length polarization can be characterized via:
\begin{fact}[{\cite{cho_infinitesimal_2020}}]\label{fact:reduction1}
	For a Darboux pair $f_0, f_1 : (I, \frac{\dif s^2}{m}) \to \mathbb{C}$ with paramter $\mu$, assume that $f_0$ is arc-length polarized.
	Then $f_1$ is also arc-length polarized if and only if $|f_1 - f_0|^2 = \frac{1}{\mu}$ at one point $s_0 \in I$.
\end{fact}

The system of semi-discrete isothermic surfaces, with one smooth parameter representing the parametrization of the curves and one discrete parameter describing the transformations, was given a new interpretation in \cite{cho_infinitesimal_2020} as \emph{Darboux deformations} of \emph{discrete polarized curves}, which we now recall here.

Given a discrete interval $\Sigma \subset \mathbb{Z}$, and a strictly positive or negative function $\mu$ on (unoriented) edges of $\Sigma$, a \emph{discrete polarized curve} is a discrete curve $x:(\Sigma,\tfrac{1}{\mu}) \to \mathbb{C}$ defined on a discrete polarized domain $(\Sigma,\tfrac{1}{\mu})$ whose polarization is given by $\mu$.

\begin{remark}
	The discrete polarization $\frac{1}{\mu}$ is a straight discretization of the non-vanishing quadratic differential $\frac{\dif s^2}{m}$ in the smooth case; hence, the discrete polarization is assumed to be strictly positive or negative. This is akin to the a priori assumption of the umbilic-free condition when studying the local properties of isothermic surfaces.
\end{remark}

\begin{definition}[{\cite{cho_infinitesimal_2020}}]
	A continuous motion of discrete curves $f : \Sigma \times I \to \mathbb{C}$ is called an \emph{infinitesimal Darboux transformation}, or a \emph{Darboux deformation}, with parameter function $m$ of a discrete polarized curve $x:(\Sigma,\tfrac{1}{\mu}) \to \mathbb{C}$ if  $f(s_0) = x$ for some $s_0 \in I$, and on every edge $(ij)$,
	\[
		\frac{f'_i f'_j}{(f_i - f_j)^2} = \frac{\mu_{ij}}{m}.
	\]
\end{definition}
It is readily seen that $f$ is, in fact, a semi-discrete isothermic surface.
One can also consider a discrete arc-length polarization:
\begin{definition}[{\cite{cho_infinitesimal_2020}}]
	A discrete curve $x :(\Sigma,\tfrac{1}{\mu}) \to \mathbb{C}$ is arc-length polarized if, on every edge $(ij)$,
	\[
		|x_i - x_j|^2 = \frac{1}{\mu_{ij}}.
	\]
\end{definition}
Furthermore, the condition of Darboux deformation keeping the arc-length polarization was identified:
\begin{fact}[{\cite{cho_infinitesimal_2020}}]\label{fact:reduction2}
	Let $f : \Sigma \times I \to \mathbb{C}$ be a Darboux deformation of a discrete arc-length polarized curve $x:(\Sigma,\tfrac{1}{\mu}) \to \mathbb{C}$ with parameter function $m$ so that $f(s_0) = x$ for some $s_0 \in I$.
	Then $f(s_1):(\Sigma,\tfrac{1}{\mu}) \to \mathbb{C}$ is arc-length polarized for all $s_1 \in I$ if and only if $f_i : (I, \frac{\dif s^2}{m}) \to \mathbb{C}$ is (smooth) arc-length polarized for some $i \in \Sigma$.
\end{fact}

Finally, the semi-discrete potential mKdV equation was obtained via Darboux deformations of discrete polarized curves keeping the arc-length condition, with $'$ denoting the differentiation with respect to $s$:
\begin{fact}[{\cite{cho_infinitesimal_2020}}]\label{fact:ddsemi}
	Let $f : \Sigma \times I \to \mathbb{C}$ be a Darboux deformation of a discrete arc-length polarized curve $x:(\Sigma,\tfrac{1}{\mu}) \to \mathbb{C}$ with parameter function $m$ keeping the arc-length polarization.
	Assuming without loss of generality that $m \equiv 1$, and defining $\theta_i : I \to \mathbb{R}$ as the tangential angle of $f_i : (I, \frac{\dif s^2}{m}) \to \mathbb{C}$ for all $i \in \Sigma$, that is, $f_i' = e^{\sqrt{-1} \theta_i}$, $\theta$ becomes the solution to the semi-discrete potential mKdV equation:
	\[
		\left(\frac{\theta_i + \theta_j}{2}\right)' = \frac{2}{\sqrt{\mu_{ij}}} \sin \left(\frac{\theta_j - \theta_i}{2}\right)
	\]
	on any edge $(ij)$.
\end{fact}

\begin{remark}
	The Darboux deformation described in Fact~\ref{fact:ddsemi} keeps the length of the edges; therefore, this deformation can be viewed as certain isoperimetric deformation of a discrete curve (see, for example, \cite{inoguchi_explicit_2012}).
\end{remark}

\section{Discrete potential mKdV equation via permutability}

\subsection{Permutability between Darboux transformations and Darboux deformations}
With the goal of interpreting the $3$-dimensional system of successive Darboux transformations of semi-discrete isothermic surfaces as permutability between Darboux transformations and Darboux deformations of discrete polarized curves, let us first explicitly state the definition of Darboux transformations for discrete polarized curves.

\begin{definition}\label{def:dFlow}
	Two discrete polarized curves $x, \hat{x}: (\Sigma, \frac{1}{\mu}) \to \mathbb{C}$ are called a \emph{Darboux pair with parameter $\hat{\mu}$} if, on every edge $(ij)$,
		\begin{equation}\label{eqn:dFlow}
			\cratio(x_i, x_j, \hat{x}_j, \hat{x}_i) = \frac{x_i - x_j}{x_j - \hat{x}_j} \frac{\hat{x}_j - \hat{x}_i}{\hat{x}_i - x_i} = \frac{\hat{\mu}}{\mu_{ij}},
		\end{equation}
	for some non-zero constant $\hat{\mu} \in \mathbb{R}\setminus \{0\}$. We call one of the curves a \emph{Darboux transform} of the other.
\end{definition}

\begin{remark}
	The condition \eqref{eqn:dFlow} for Darboux transformations of discrete polarized curves is identical to the well-known characterisation of discrete isothermic surfaces, first introduced in \cite{bobenko_discrete_1996-1}, since successive Darboux transformations of discrete polarized curve should yield discrete isothermic surfaces, analogous to the semi-discrete case of \cite{burstall_semi-discrete_2016}.
\end{remark}

Note that a Darboux transformation is determined by the choice of the parameter $\hat{\mu}$ and an initial condition $\hat{x}_i$ at some vertex $i \in \Sigma$.

\begin{figure}
	\centering
	\savebox{\mybox}{\includegraphics[width=0.45\textwidth]{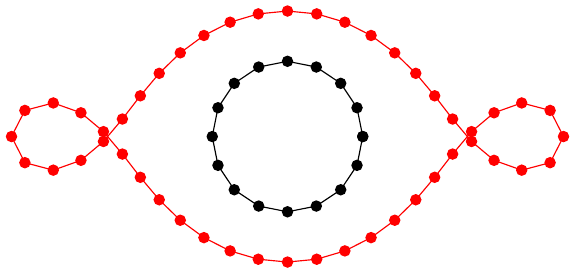}}
	\begin{subfigure}{0.45\textwidth}
		\vbox to \ht\mybox{%
			\vfill
			\includegraphics[width=1\textwidth]{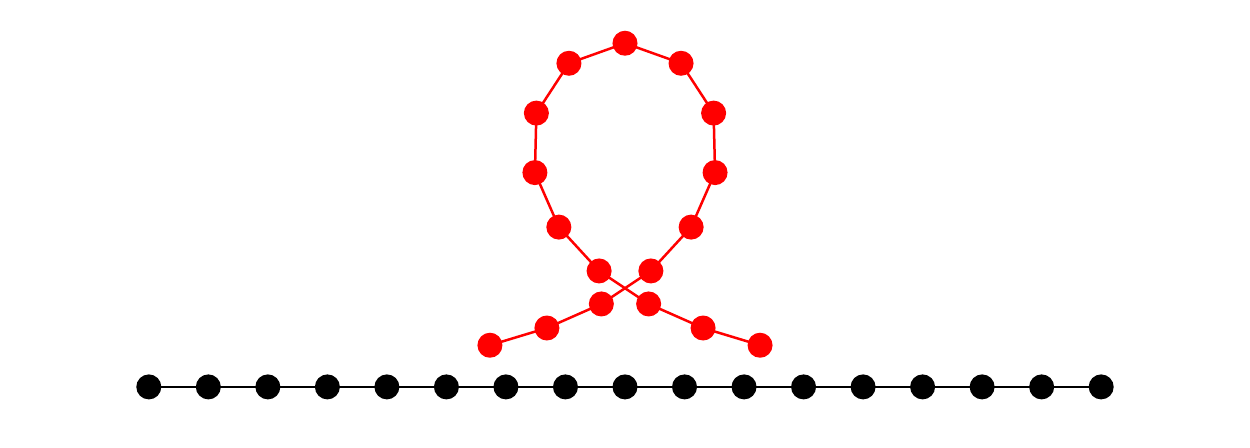}
			\vfill
		}
		\caption{}
	\end{subfigure}
	~
	\begin{subfigure}{0.45\textwidth}
		\usebox{\mybox}
		\caption{}
	\end{subfigure}

	\par\bigskip
	\savebox{\mybox}{\includegraphics[width=0.45\textwidth]{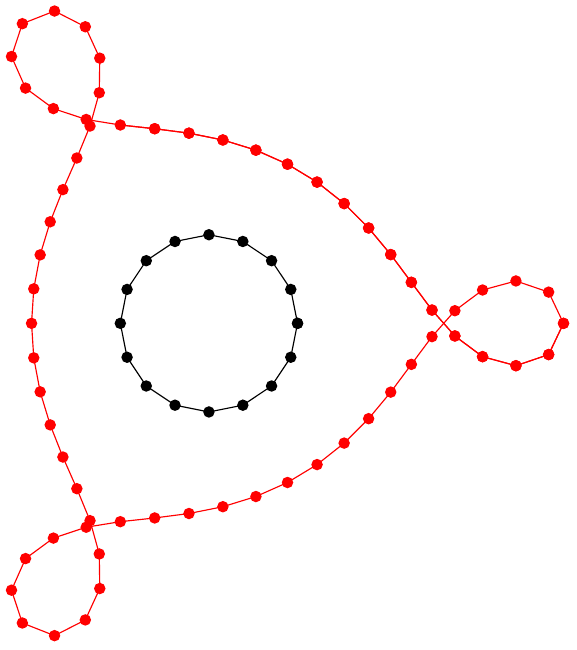}}
	\begin{subfigure}{0.45\textwidth}
		\usebox{\mybox}
		\caption{}
	\end{subfigure}
	~
	\begin{subfigure}{0.45\textwidth}
		\vbox to \ht\mybox{%
			\vfill
			\includegraphics[width=1\textwidth]{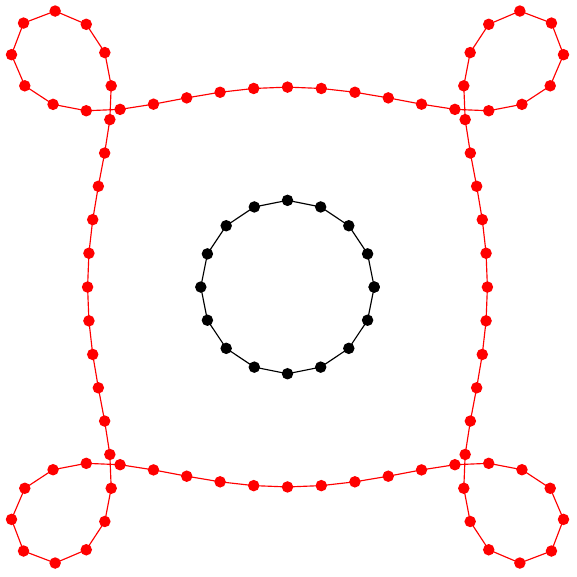}
			\vfill
		}
		\caption{}
	\end{subfigure}
	\caption{Examples of Darboux transforms (red curves) of a given curve (black curves) keeping the arc-length polarization, determined via Equation \eqref{def:dFlow}. In particular, figures (B), (C), (D) show Darboux transforms (red curve) of a triply, quadruply, quintuply covered discrete circle (black circle), respectively.}
	\label{fig:examples1}
\end{figure}

Now we identify the condition for the Darboux transform of an arc-length polarized discrete curve to be arc-length polarized again, with some examples shown in Figure~\ref{fig:examples1}.
\begin{proposition}\label{prop:discBy} 
	Let $x, \hat{x} : (\Sigma, \frac{1}{\mu}) \to \mathbb{C}$ be a Darboux pair with parameter $\hat{\mu}$ where $x$ is arc-length polarized.
	Then $\hat{x}$ is also arc-length polarized if and only if $|x_i - \hat{x}_i|^2 = \frac{1}{\hat{\mu}}$ at some vertex $i \in \Sigma$.
\end{proposition}

\begin{proof}
Assume first that $|x_i - \hat{x}_i|^2 = \frac{1}{\hat{\mu}}$ at some vertex $i \in \Sigma$. 
Then on an edge $(ij)$, the definition of Darboux pair with parameter $\hat{\mu}$ \eqref{eqn:dFlow} tells us 
  \[
      \hat{x}_j-\hat{x}_i=\frac{\hat{\mu}(\hat{x}_i-x_i)(x_j-\hat{x}_i)}
                                       {\mu_{ij}(x_i-x_j)+\hat{\mu}(\hat{x}_i-x_i)}.
  \]
Then a computation gives 
  \[
      |\mu_{ij}(x_i-x_j)+\hat{\mu}(\hat{x}_i-x_i)|^2=\mu_{ij} \hat{\mu}|x_j-\hat{x}_i|^2;
  \]
hence, 
  \[
      |\hat{x}_j-\hat{x}_i|^2
                                  =\frac{\hat{\mu}|x_j-\hat{x}_i|^2}{\mu_{ij} \hat{\mu}|x_j-\hat{x}_i|^2}
                                  =\frac{1}{\mu_{ij}}.
  \]
We can prove the converse claim similarly, by switching the roles of $\hat{x}_j$ and $x_i$.
\end{proof}

\begin{remark}
	Note that now for arc-length polarization preserving Darboux transformations, the parameter $\hat{\mu}$ is a positive constant.
\end{remark}

\begin{remark}\label{rem:omg}
	If two arc-length polarized discrete curves $x, \hat{x} : (\Sigma, \tfrac{1}{\mu}) \to \mathbb{C}$ are a Darboux pair with parameter $\hat{\mu}$, then the cross-ratio condition \eqref{eqn:dFlow} and Proposition~\ref{prop:discBy} immediately implies that $|x_i - \hat{x}_i|^2 = \frac{1}{\hat{\mu}}$ for all $i \in \Sigma$.
\end{remark}

We now discuss the permutability between infinitesimal Darboux transformation and Darboux transformation of a discrete polarized curve.
Let two discrete polarized curves $x, \hat{x} : (\Sigma, \tfrac{1}{\mu}) \to \mathbb{C}$ be a Darboux pair with parameter $\hat{\mu}$, and further assume that $f : \Sigma \times I \to \mathbb{C}$ is a Darboux deformation of $x$ with parameter function $m$, so that $f(s_0) = x$ for some fixed $s_0 \in I$.
Then we have that the two smooth polarized curves $f_i, f_j : (I, \tfrac{\dif s^2}{m}) \to \mathbb{C}$ are a Darboux pair with parameter $\mu_{ij}$ on every edge $(ij)$  (see Figure~\ref{fig:perm1}).

For some fixed $i \in \Sigma$, let $\hat{f}_i : (I, \tfrac{\dif s^2}{m}) \to \mathbb{C}$ be a Darboux transform of $f_i$ with parameter $\hat{\mu}$, determined uniquely by taking the initial condition as $\hat{f}_i(s_0) = \hat{x}_i$  (see Figure~\ref{fig:perm2}).
Therefore, the permutability of Darboux transformations (see Fact~\ref{fact:perm}) of smooth polarized curves ensures the existence of $\hat{f}_j$ that is simultaneously a Darboux transform of $f_j$ and $\hat{f}_i$ with parameters $\hat{\mu}$ and $\mu_{ij}$, respectively (see Figure~\ref{fig:perm3}).
Defining $\hat{f} : \Sigma \times I \to \mathbb{C}$ recursively using permutability (see Figure~\ref{fig:perm4}), the cross-ratio condition for the permutability \eqref{eqn:crsmooth} implies that
\begin{enumerate}
	\item $\hat{f}(s_0) = \hat{x}$, giving us that $\hat{f}$ is a Darboux deformation of $\hat{x}$ with parameter function $m$, and
	\item $f(s_1)$ and $\hat{f}(s_1)$ are a Darboux pair with parameter $\hat{\mu}$ for all $s_1 \in I$.
\end{enumerate}

\begin{figure}
	\centering
	\savebox{\mybox}{\includegraphics[width=0.48\textwidth]{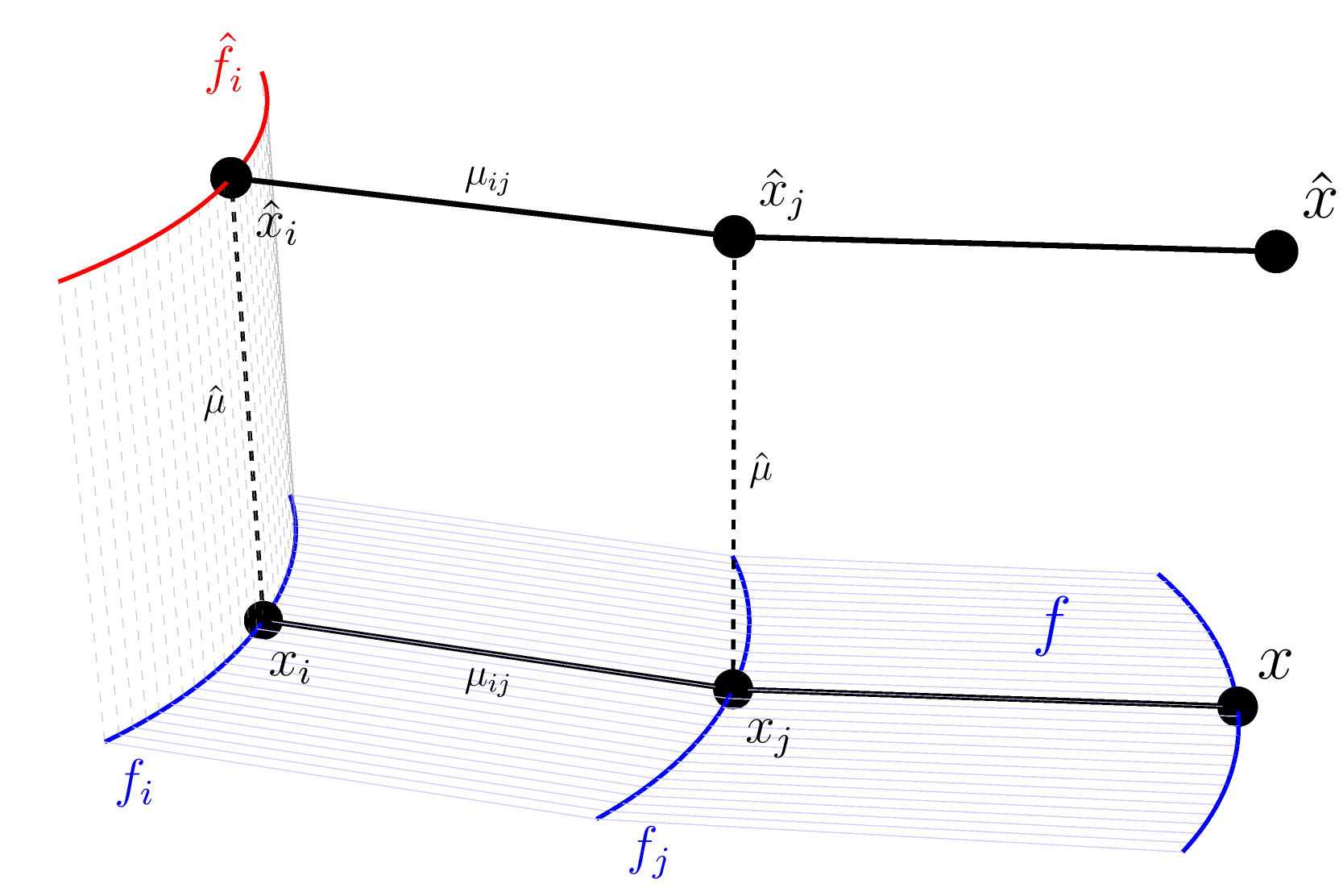}}
	\begin{subfigure}{0.48\textwidth}
		\vbox to \ht\mybox{%
			\vfill
			\includegraphics[width=1\textwidth]{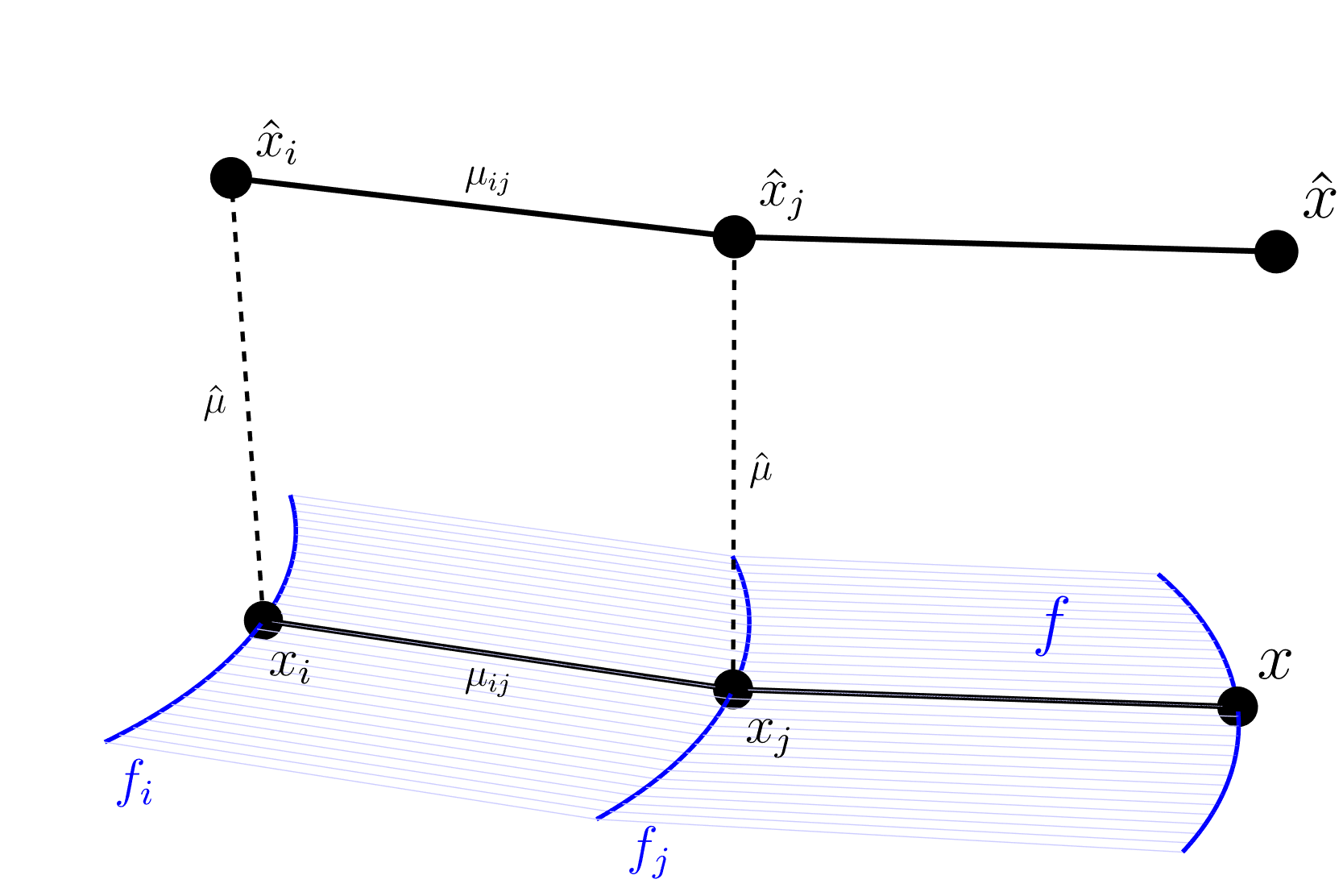}
			\vfill
		}
		\caption{}
		\label{fig:perm1}
	\end{subfigure}
	~
	\begin{subfigure}{0.48\textwidth}
		\usebox{\mybox}
		\caption{}
		\label{fig:perm2}
	\end{subfigure}

	\par\bigskip
	\savebox{\mybox}{\includegraphics[width=0.48\textwidth]{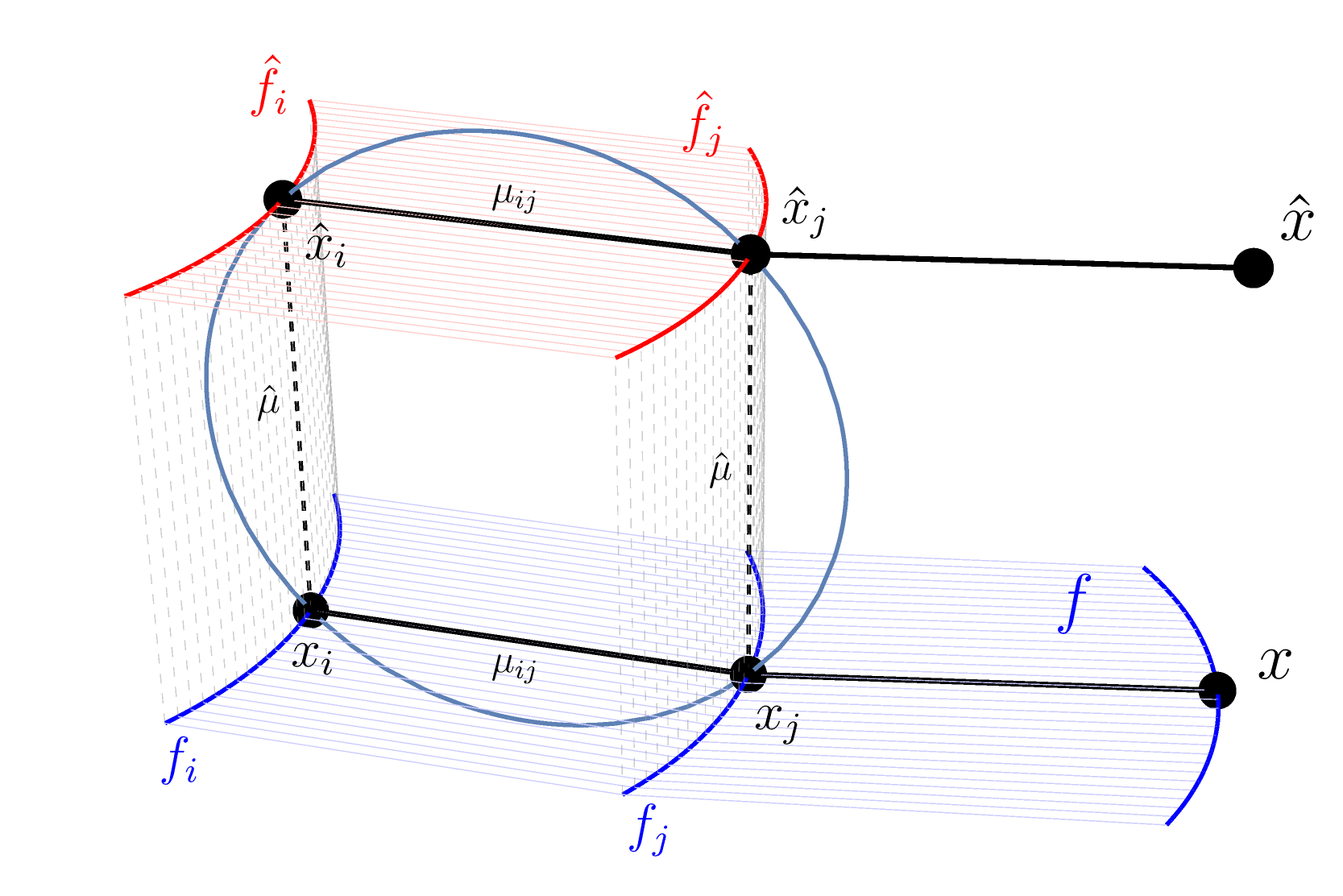}}
	\begin{subfigure}{0.48\textwidth}
		\usebox{\mybox}
		\caption{}
		\label{fig:perm3}
	\end{subfigure}
	~
	\begin{subfigure}{0.48\textwidth}
		\vbox to \ht\mybox{%
			\vfill
			\includegraphics[width=1\textwidth]{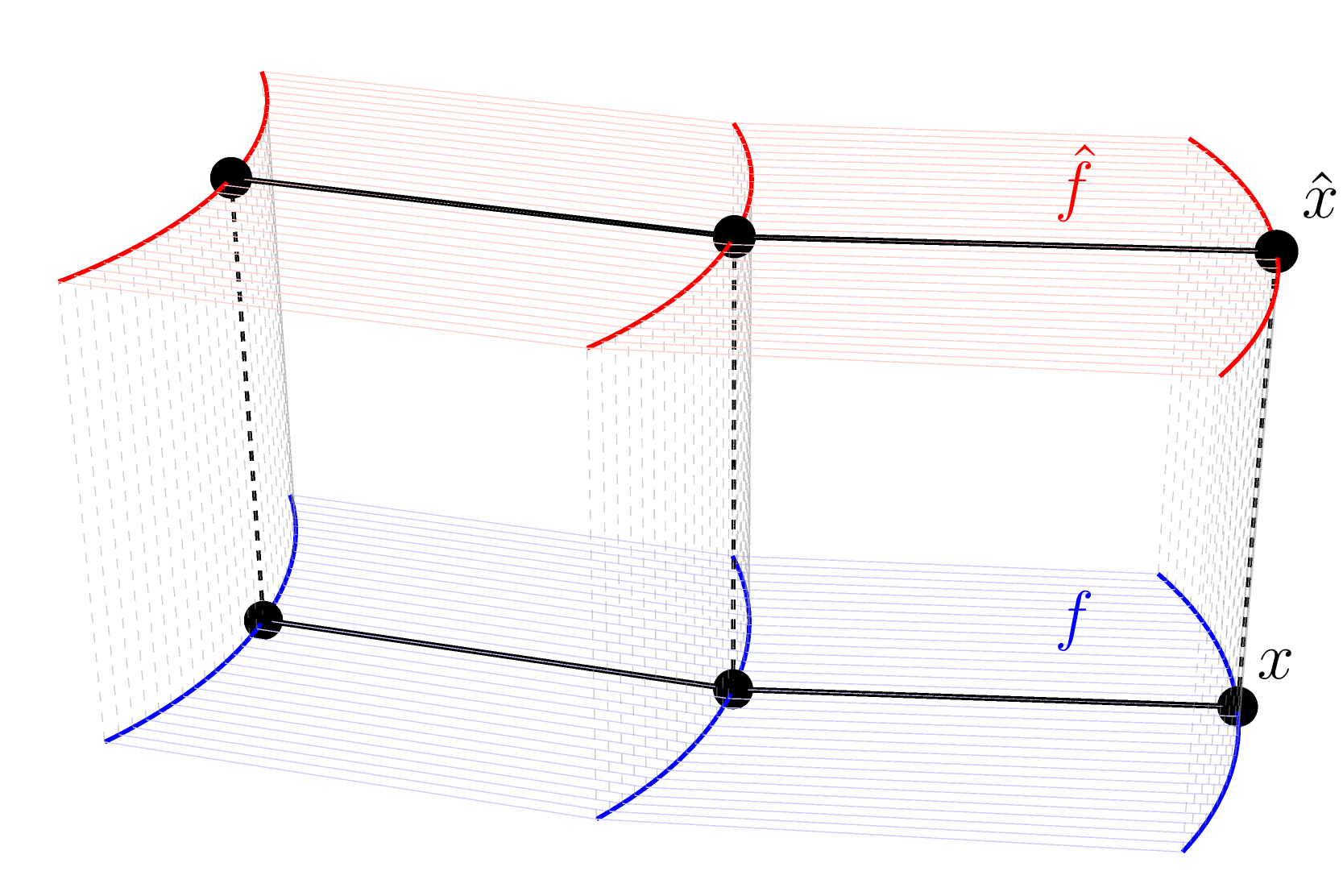}
			\vfill
		}
		\caption{}
		\label{fig:perm4}
	\end{subfigure}
	\caption{Conceptual diagram explaining the permutability between Darboux transformation and Darboux deformation of discrete polarized curves coming from Darboux transformation of semi-discrete isothermic surfaces.}
	\label{fig:examples}
\end{figure}

Thus, we have the permutability between Darboux transformation and Darboux deformation of a discrete polarized curve coming directly from the Darboux transformation of semi-discrete isothermic surfaces:
\begin{proposition}\label{infDarbPerm}
	Let $f:\Sigma \times I \rightarrow \mathbb{C}$ be a Darboux deformation of a discrete polarized curve $x:(\Sigma,\tfrac{1}{\mu}) \rightarrow \mathbb{C}$ with parameter function $m$ and let $\hat{x} $ be a Darboux transform of $x$ with parameter $\hat{\mu}$. Then there is a Darboux deformation $\hat{f}$ of $\hat{x}$ with parameter function $m$ so that, for any $s \in I$, $\hat{f}(s)$ is a Darboux transform of $f(s)$ with parameter $\hat{\mu}$.
\end{proposition}

In such case of $f$ and $\hat{f}$ in Proposition~\ref{infDarbPerm}, we say that $f$ and $\hat{f}$ form \emph{infinitesimal Bianchi cubes}.

In fact, the permutability holds even with the extra condition of keeping the arc-length polarization:
\begin{proposition} \label{PermOfDarb}
	Let two arc-length polarized discrete curves $x, \hat{x} : (\Sigma,\tfrac{1}{\mu}) \rightarrow \mathbb{C}$ be a Darboux pair with parameter $\hat{\mu}$, and further assume that $f, \hat{f}: \Sigma \times I \rightarrow \mathbb{C}$ are Darboux deformations of $x, \hat{x}$ with parameter function $m$, respectively, forming infinitesimal Bianchi cubes.
	If $f$ is a Darboux deformation keeping the arc-length polarization, then so is $\hat{f}$.
\end{proposition}

\begin{proof}
	Note that we only need to show that the smooth curve $\hat{f}_i$ is arc-length polarized for some $i \in \Sigma$ by Fact~\ref{fact:reduction2}.
	For this, note that on any vertex $i$, since $f$ is a Darboux deformation keeping arc-length polarization, we have that the smooth curve $f_i$ is arc-length polarized by $\tfrac{\dif s^2}{m}$, again by Fact~\ref{fact:reduction2}.
	Fixing $s_0 \in I$ such that $f(s_0) = x$ and $\hat{f}(s_0) = \hat{x}$, Remark~\ref{rem:omg} then implies that 
	\begin{equation}\label{eqn:what}
		|\hat{f}_i(s_0) - f_i(s_0)|^2 = |\hat{x}_i - x_i|^2 = \frac{1}{\hat{\mu}}
	\end{equation}
for all $i \in \Sigma$.
	However, since the two smooth curves $f_i$ and $\hat{f}_i$ are a Darboux pair with parameter $\hat{\mu}$, and $f_i$ is arc-length polarized, \eqref{eqn:what} implies that $\hat{f}_i$ must also be arc-length polarized by Fact~\ref{fact:reduction1}, giving us the desired conclusion.
\end{proof}

\subsection{Discrete potential mKdV equation coming from infinitesimal Bianchi cubes}
Now we prove that infinitesimal Bianchi cubes of Darboux transformations keeping the arc-length condition yields solutions $\vartheta :\Sigma \times \tilde \Sigma \to \mathbb{R}$, denoted by $\vartheta(n,k) = \vartheta_n^k$, to the discrete potential mKdV equation \cite{hirota_nonlinear_1977-1, hirota_discretization_1998}:
	\[
		\tan\left( \frac{\vartheta_{n+1}^{k+1} -\vartheta_n^k}{4} \right) = \frac{b_k + a_n}{b_k - a_n} \tan \left( \frac{\vartheta_n^{k+1} - \vartheta_{n+1}^k}{4} \right)
	\]
for some discrete functions $a_n : \Sigma \to \mathbb{R}$ and $b_k : \tilde\Sigma \to \mathbb{R}$.

Assume that $f, \hat{f}: \Sigma \times I \rightarrow \mathbb{C}$ are Darboux deformations (with parameter function $m$) of the Darboux pair $x, \hat{x} : (\Sigma,\tfrac{1}{\mu}) \rightarrow \mathbb{C}$ forming infinitesimal Bianchi cubes keeping the arc-length polarization.
Then on any edge $(ij)$, $(f_i, f_j)$ and $(\hat{f}_i, \hat{f}_j)$ are Darboux pairs with parameter $\mu_{ij}$ keeping the arc-length polarization.
Without loss of generality, we take $m \equiv 1$; hence, defining $\theta_i, \theta_j, \hat{\theta}_i, \hat{\theta}_j$ as tangential angles of $f_i, f_j, \hat{f}_i, \hat{f}_j$, respectively, we obtain a pair of semi-discrete potential mKdV equations
	\begin{equation}\label{eqn:semi1}
		\left(\frac{\theta_i + \theta_j}{2}\right)' = \frac{2}{\sqrt{\mu_{ij}}} \sin \left(\frac{\theta_j - \theta_i}{2}\right) \quad\text{and}\quad \left(\frac{\hat{\theta}_i + \hat{\theta}_j}{2}\right)' = \frac{2}{\sqrt{\mu_{ij}}} \sin \left(\frac{\hat{\theta}_j - \hat{\theta}_i}{2}\right).
	\end{equation}
Similarly, we have that $(f_i, \hat{f}_i)$ and $(f_j, \hat{f}_j)$ are Darboux pairs with parameter $\hat{\mu}$ keeping the arc-length polarization, and we obtain another pair of semi-discrete potential mKdV equations
	\begin{equation}\label{eqn:semi2}
		\left(\frac{\theta_i + \hat{\theta_i}}{2}\right)' = \frac{2}{\sqrt{\hat{\mu}}} \sin \left(\frac{\hat{\theta}_i - \theta_i}{2}\right) \quad\text{and}\quad \left(\frac{\theta_j + \hat{\theta}_j}{2}\right)' = \frac{2}{\sqrt{\hat{\mu}}} \sin \left(\frac{\hat{\theta}_j - \theta_j}{2}\right).
	\end{equation}

These equations \eqref{eqn:semi1} and \eqref{eqn:semi2} are well-known partial differential equations that define Bäcklund transformations of the smooth potential mKdV equation as seen in \cite[Equations (7), (8)]{wadati_backlund_1974}.
Using these equations, permutability of the transformation was obtained in \cite[Equation (9)]{wadati_backlund_1974} (see also \cite{bianchi_sulla_1892}):
	\begin{equation}\label{eqn:discretem}
		\tan\left( \frac{\hat{\theta}_j - \theta_i}{4} \right) = \frac{\sqrt{\hat{\mu}} + \sqrt{\mu_{ij}}}{\sqrt{\hat{\mu}} - \sqrt{\mu_{ij}}} \tan \left( \frac{\hat{\theta}_i - \theta_j}{4} \right),
	\end{equation}
which is the discrete potential mKdV equation.
Summarizing, we have:

\begin{theorem}\label{thm1}
	The infinitesimal Bianchi cubes of Darboux transformations keeping arc-length polarization in both the smooth and discrete directions yield 1-parameter families of solutions of the discrete potential mKdV equation.
\end{theorem}

\begin{remark}
	We remark here that obtaining \eqref{eqn:discretem} from \eqref{eqn:semi1} and \eqref{eqn:semi2} is a well-known calculation, which we do not reproduce here.
	The original calculation by Bianchi can be found in \cite{bianchi_sulla_1892}, where Equations (9) and (9*) correspond to \eqref{eqn:semi1} and \eqref{eqn:semi2}, and Equation (11) corresponds to \eqref{eqn:discretem}.
	Another detailed (modern) calculation can be found in \cite[\S 1.3.1]{rogers_backlund_2002}.
\end{remark}

Finally, we make connection to previous work \cite[\S 2.2]{matsuura_discrete_2012-1} that considered a discrete motion of discrete curves whose compatibility condition results in the discrete potential mKdV equation.
As in \cite{matsuura_discrete_2012-1}, let $x:\Sigma \times \tilde\Sigma \to \mathbb{C}$, denoted by $x(n, k) = x_n^k$, be a discrete motion of a discrete planar curve such that, for any $(n, k) \in  \Sigma \times \tilde\Sigma$,
	\begin{equation}\label{eqn:discreteMotion}
		|x^k_{n+1}-x^k_n|=:a_n \quad\text{and}\quad |x^{k+1}_n-x^k_n|=:b_k
	\end{equation}
are constant in $k$ and $n$, respectively.
 
Without loss of generality, assume $x^k_n = 0$, $x^k_{n+1} = a_n$ and $x^{k+1}_n = b_k e^{i \theta}$ for some $\theta \in \mathbb{R}$.
Excluding the solution $a_n + b_k e^{i \theta}$ obtained via translation, $x^{k+1}_{n+1}$ is uniquely determined with cross-ratio satisfying
	\[
		\cratio{(x^k_n, x^k_{n+1}, x^{k+1}_{n+1}, x^{k+1}_n)}=\frac{{a_n}^2}{{b_k}^2}.
	\]
Prescribing the discrete arc-length polarization on $x_n^k, x_n^{k+1} : (\Sigma, \tfrac{1}{a_n^2}) \to \mathbb{C}$, we see that the two discrete polarized curves $x_n^k, x_n^{k+1}$ are a Darboux pair with parameter $b_k^2$.
Hence, the discrete motion considered in \cite{matsuura_discrete_2012-1} is successive Darboux transformations of discrete polarized curves keeping the arc-length polarization.

\begin{remark}
	As a final remark, we note that the Darboux transformations of discrete polarized curves keeping the arc-length condition requires the defining cross-ratios \eqref{eqn:dFlow} to be positive, resulting in non-embedded quadrilaterals.
	This phenomenon is also observed in consideration of Darboux transformations of smooth polarized curves keeping the arc-length condition, as now the tangential cross-ratios are positive.
	This non-embeddedness seems to be a result of the freedom created by the discrete or semi-discrete systems, raising the question of what the continuum limit of these situations might correspond to in the smooth case, and what the possible relations are to the smooth motions of smooth curves resulting in smooth potential mKdV equations studied in works such as \cite{lamb_solitons_1976, goldstein_korteweg-vries_1991}.
\end{remark}

\vspace{15pt}
\textbf{Acknowledgements.}
The authors would like to thank Professor Udo Hertrich-Jeromin and the referees for valuable comments.
The authors gratefully acknowledge the support from the JSPS/FWF Bilateral Joint Project I3809-N32 ``Geometric shape generation" and JSPS Grants-in-Aid for: JSPS Fellows 19J10679, Scientific Research (C) 15K04845, (C) 20K03585 and (S) 17H06127 (P.I.: M.-H.\ Saito).


\end{document}